\documentclass[11pt]{article}
\usepackage[margin=1in]{geometry} 
\usepackage{amsmath}
\usepackage{amssymb,amsfonts,bbm,mathrsfs,stmaryrd}
\usepackage{float}
\usepackage{relsize}
\usepackage{indentfirst}

\usepackage[colorlinks,
             linkcolor=black!50!red,
             citecolor=blue,
             pdftitle={Completeness of the Six Vertex Model with Reflecting Boundary Conditions},
             pdfauthor={Ammar Husain},
             pdfsubject={Spin Chains},
             pdfkeywords={spin chains, quantum mechanics,XXZ},
             pdfproducer={pdfLaTeX},
             pdfpagemode=None,
             bookmarksopen=true
             bookmarksnumbered=true]{hyperref}

\usepackage[amsmath,thmmarks,hyperref]{ntheorem}
\usepackage{cleveref}

\theoremstyle{change}

\newtheorem{defn}[equation]{Definition}

\newtheorem{theorem}[equation]{Theorem}

\newtheorem{prop}[equation]{Proposition}
\newtheorem{lemma}[equation]{Lemma}

\theorembodyfont{\upshape}
\theoremsymbol{\ensuremath{\Diamond}}

\theoremstyle{nonumberplain}

\theoremsymbol{\ensuremath{\Box}}
\newtheorem{proof}{Proof}

\qedsymbol{\ensuremath{\Box}}

\creflabelformat{equation}{#2(#1)#3} 

\crefname{equation}{equation}{equations}
\crefname{eg}{example}{examples}
\crefname{defn}{definition}{definitions}
\crefname{prop}{proposition}{propositions}
\crefname{thm}{Theorem}{Theorems}
\crefname{lemma}{lemma}{lemmas}
\crefname{cor}{corollary}{corollaries}
\crefname{remark}{remark}{remarks}
\crefname{section}{Section}{Sections}
\crefname{subsection}{Section}{Sections}

\crefformat{equation}{#2equation~(#1)#3} 
\crefformat{eg}{#2example~#1#3} 
\crefformat{defn}{#2definition~#1#3} 
\crefformat{prop}{#2proposition~#1#3} 
\crefformat{thm}{#2Theorem~#1#3} 
\crefformat{lemma}{#2lemma~#1#3} 
\crefformat{cor}{#2corollary~#1#3} 
\crefformat{remark}{#2remark~#1#3} 
\crefformat{section}{#2Section~#1#3} 
\crefformat{subsection}{#2Section~#1#3} 

\Crefformat{equation}{#2Equation~(#1)#3} 
\Crefformat{eg}{#2Example~#1#3} 
\Crefformat{defn}{#2Definition~#1#3} 
\Crefformat{prop}{#2Proposition~#1#3} 
\Crefformat{thm}{#2Theorem~#1#3} 
\Crefformat{lemma}{#2Lemma~#1#3} 
\Crefformat{cor}{#2Corollary~#1#3} 
\Crefformat{remark}{#2Remark~#1#3} 
\Crefformat{section}{#2Section~#1#3} 
\Crefformat{subsection}{#2Section~#1#3}

\numberwithin{equation}{section}

\usepackage{tikz}
\tikzset{dot/.style={circle,draw,fill,inner sep=1pt}}
\usetikzlibrary{decorations.markings}
\usepackage{braids}
\usetikzlibrary{chains,matrix,scopes,decorations,shapes,arrows}
\usetikzlibrary{cd}

\newcommand\ket[1]{\mid #1 \rangle}
\newcommand\setof[1]{\{ #1 \}}
\newcommand\lt{<}
\newcommand\abs[1]{ \mid #1 \mid }

\title{Completeness of the Six Vertex Model with Reflecting Boundary Conditions}
\author{Ammar Husain \thanks{Electronic address: \texttt{ahusain@berkeley.edu}}}
\date{}

\begin{document}
\maketitle

\tikzset{->-/.style={decoration={
  markings,
  mark=at position #1 with {\arrow{>}}},postaction={decorate}}}
\tikzset{-<-/.style={decoration={
  markings,
  mark=at position #1 with {\arrow{<}}},postaction={decorate}}}

\begin{abstract}

In this note, we prove the completeness of Bethe vectors for the six vertex model with diagonal reflecting boundary conditions. We show that as inhomogeneity parameters get sent to infinity in a successive order the Bethe vectors give a complete basis of the space of states.

\end{abstract}

\section{Introduction}

The spectrum of quantum spin chains and properties of integrability are sensitive to boundary conditions. A characterization of such boundary conditions comes up in Cherednik \cite{Cherednik84} and Sklyanin \cite{Sklyanin88} .\par

In this paper we prove that the set of eigenvectors constructed in \cite{Sklyanin88} is asymptotically complete. We consider an inhomogeneous spin chain with inhomogeneities in a sector $t_1 \cdots t_N$ following $0 < Re(t_1) < \cdots << Re(t_N)$. The Bethe vectors for such a spin chain form a basis for the entire Hilbert space. \par

The plan of this paper is as follows. In \cref{section 2} we recall facts about the inhomogeneous six-vertex model with reflecting boundaries. In \cref{section3} we describe the asymptotics in this sector to solutions of Bethe equations. The Bethe vectors and the proof of completeness is contained in \cref{section4}.

\section{The Six Vertex Model with Reflecting Boundary} \label{section 2}

\subsection{Notation}
The Boltzmann weights are parameterized with the R matrix
\begin{equation*} \label{RMatrix}
R = \begin{pmatrix}
b(x+\eta)&0&0&0\\
0&b(x)&b(\eta )&0\\
0&b(\eta )&b(x)&0\\
0&0&0&b(x+\eta )\\
\end{pmatrix}
\end{equation*}
where $b(x) = \text{sinh} \; x$, $z=e^x$, $a_i = e^{t_i}$ and $q=e^\eta$\\

The R matrix satisfies the Yang Baxter Equation. The Reflection matrix K must satisfy the reflection equation.
\begin{eqnarray*}
R_{12} (u) R_{13} (u+v) R_{23}(v) &=& R_{23} (v) R_{13} (u+v) R_{12} (u)\\
R_{12} (u-v) K_1 (u) R_{21} (u+v) K_2 (v) &=& K_2 (v) R_{12} (u+v) K_1 (u) R_{21} (u-v)
\end{eqnarray*}

Assuming that K is diagonal leads to the one parameter family of solutions.
\begin{equation*} \label{KMatrix}
K = \begin{pmatrix}
b(x+ \xi )&0\\
0&-b(x-\xi )\\
\end{pmatrix}
\end{equation*}

\begin{defn}[Boundary Monodromy Matrices].\\

The monodromy matrix for a single row with inhomegeneities $t_1 \cdots t_N$ on those respective columns is given by\\
\begin{equation*} \label{SingleRow}
T(x, t_1 \cdots t_N ) = R_{0N} (x-t_N ) \cdots R_{01} (x-t_1 ) = 
\begin{pmatrix}
A(x)&B(x)\\
C(x)&D(x)\\
\end{pmatrix}
\end{equation*}

The double row monodromy matrix takes into effect the reflection at one end. It is defined as

\begin{equation*} \label{DoubleRow}
U(x,t_1 \cdots t_N, \xi_+ ) = T(x,\vec{t} ) K ( x - \frac{\eta}{2} , \xi_+ ) \sigma_2 T(-x,\vec{t} ) \sigma_2 =
\begin{pmatrix}
\mathcal{A}(x)&\mathcal{B}(x)\\
\mathcal{C}(x)&\mathcal{D}(x)\\
\end{pmatrix}
\end{equation*}

\end{defn}

\subsection{Commutation Relations and Bethe Ansatz}

Sklyanin proved the reflection equation which implies the following commutation relations for the operator valued entries of the 2 by 2 double row mondodromy matrix. \cite{Sklyanin88}
\begin{eqnarray*}
\mathcal{A} (u) \mathcal{B} (v) &=& \frac{b(u-v-\eta ) b( u + v - \eta )}{b(u-v) b(u+v)} \mathcal{B} (v) \mathcal{A} (u)\\
&+& \frac{b(\eta) b(u+v-\eta )}{b(u-v) b(u+v) } \mathcal{B} (u) \mathcal{A} (v) - \frac{b(\eta)}{b(u+v)} \mathcal{B} (u) \mathcal{D}(v)\\
\mathcal{D} (u) \mathcal{B} (v) &=& \frac{b(u-v+\eta ) b( u + v + \eta )}{b(u-v) b(u+v)} \mathcal{B} (v) \mathcal{D} (u) - \frac{b(2 \eta) b(\eta )}{b(u-v) b(u+v) } \mathcal{B} (v) \mathcal{A} (u)\\
&+& \frac{b(\eta)b(u+v+\eta )}{b(u-v) b(u+v)} \mathcal{B} (u) \mathcal{D}(v) + \frac{b(u-v+2\eta ) b(\eta )}{b(u-v) b(u+v)} \mathcal{B} (u) \mathcal{A} (v)\\
\end{eqnarray*}

The relations are simpler if we change variables to use $\tilde{\mathcal{D}} (u) = \mathcal{D} (u) b(2u) - \mathcal{A}(u) b(\eta)$ instead of $\mathcal{D}$\\
\begin{eqnarray*}
\tilde{\mathcal{D}} (u) \mathcal{B} (v) &=& \frac{b(u-v+\eta ) b(u+v+\eta )}{b(u-v)b(u+v)}\mathcal{B} (v) \tilde{\mathcal{D}} (u) + \frac{b(\eta ) b( 2u + \eta ) b(2v - \eta )}{b(u+v) b(2v)} \mathcal{B} (u) \mathcal{A} (v) - \frac{b(\eta ) b(2u + \eta )}{b(u-v) b(2v)} \mathcal{B}(u) \tilde{\mathcal{D}} (v)
\end{eqnarray*}

The transfer matrix associated with the above monodromy matrix illustrated in Figure \ref{fig:doubletransfer} is given by
\begin{eqnarray*}
t(u, \xi_+ , \xi_- ) = tr( K(u+\frac{\eta}{2} , \xi_+ ) U_- (u)) &=& b(u+\xi_+ + \frac{\eta}{2} ) \mathcal{A} (u) - b(u-\xi_+ + \frac{\eta}{2} ) \mathcal{D} (u)\\
&=& \frac{b(2u+\eta)}{b(2u)} b(u+\xi_+ - \frac{\eta}{2} ) \mathcal{A} (u) - \frac{1}{b(2u)} b ( u-\xi_+ + \frac{\eta}{2} ) \tilde{\mathcal{D}} (u)
\end{eqnarray*}

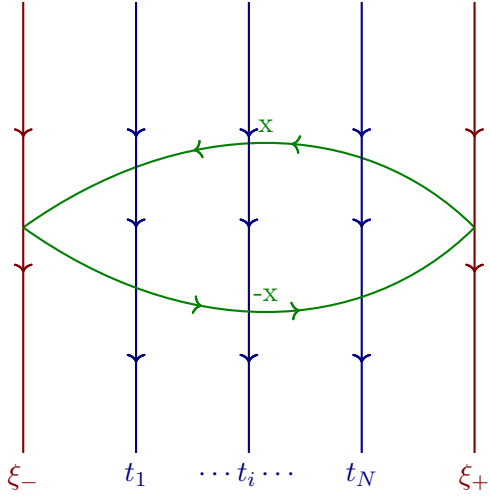
\begin{figure}[H]
\centering
 \begin{tikzpicture}[scale=.3,color=green!50!black]    
    \draw[thick,-<-=.4,-<-=.6] (0,0) .. node[above] {x} controls (7,5) and (15,5) .. (20,0);
    \path (0,10) -- (0,-10) [draw,thick,->-=.3,->-=.6,color=red!50!black] node[below] {$\xi_-$};
    \path (5,10) -- (5,-10) [draw,thick,->-=.3,->-=.5,->-=.8,color=blue!50!black] node[below] {$t_1$};
    \path (10,10) -- (10,-10) [draw,thick,->-=.3,->-=.5,->-=.8,color=blue!50!black] node[below] {$\cdots t_i \cdots$};
    \path (15,10) -- (15,-10) [draw,thick,->-=.3,->-=.5,->-=.8,color=blue!50!black] node[below] {$t_N$};
    \path (20,10) -- (20,-10) [draw,thick,->-=.3,->-=.6,color=red!50!black] node[below] {$\xi_+$};
    \draw[thick,->-=.4,->-=.6] (0,0) .. node[above] {-x} controls (7,-5) and (15,-5) .. (20,0);
  \end{tikzpicture}
\caption{The blue lines are decorated with $t_i$, The two red lines with $\xi_\pm$. Each crossing in this diagram represents a factor in the transfer matrix.} \label{fig:doubletransfer}
\end{figure}

On an off shell Bethe vector built up from the psuedovacuum $\Omega = e_-^{\otimes N}$ as $\ket{v_1 \cdots v_m} = \mathcal{B} (v_1 ) \mathcal{B} (v_2 ) \cdots \mathcal{B} (v_m) \Omega$ the result will be of the form:

\begin{align*}
t( u , \xi_+ , \xi_- ) \mathcal{B} (v_1 ) \mathcal{B} (v_2 ) \cdots \mathcal{B} (v_m)  e_-^{\otimes N} &= \Lambda(u) \mathcal{B} (v_1 ) \mathcal{B} (v_2 ) \cdots \mathcal{B} (v_m) \Omega\\
&+ \sum_{j=1}^m \Lambda_j \ket {u, v_1 \cdots \hat{v}_j \cdots v_m} \label{EigenvalueEquation} \tag{$1$}\\
\Lambda (u) = \frac{b(2u+\eta )}{b(2u)} b(u+\xi_+ - \frac{\eta}{2} ) &\Delta_+ (u) \prod_{j=1}^m \frac{b(u-v_j-\eta ) b(u + v_j - \eta )}{b(u-v_j)b(u+v_j)}\\
- \frac{1}{b(2u)} b(u-\xi_+ + \frac{\eta}{2} ) &\Delta_- (u) \prod_{j=1}^m \frac{b(u-v_j+\eta ) b(u + v_j + \eta )}{b(u-v_j)b(u+v_j)}\\
\Lambda_j &= res_{u=v_j} \Lambda(u)
\end{align*}

where $\Delta_\pm$ are the eigenvalues for $\mathcal{A}$ and $\tilde{\mathcal{D}}$ on the highest weight state respectively.\\

\begin{eqnarray*}
\Delta_+ (u) &=& b(u+\xi_- \frac{\eta}{2}) \alpha(u) \delta(-u)\\
\Delta_- (u) &=& - b(2u-\eta ) b(u - \xi_+ \frac{\eta}{2}) \alpha(-u) \delta(u)\\
\end{eqnarray*}

For the vector $\ket{v_1 \cdots v_m}$ to actually be an eigenvector of $t(u,\xi_+ , \xi_-)$ ( on shell) we must have that all the $\Lambda_j$ be zero. This happens when the $v_i$ satisfy the Bethe equations which can be realized by ensuring that the poles in $\Lambda(u)$ all cancel out.\\

From Equation \ref{EigenvalueEquation}, it is clear that for the vector constructed above to be an eigenstate, the $v_i$ need to satisfy:

\begin{eqnarray*}
\frac{b( v_m + \xi_+ - \frac{\eta}{2} ) b( v_m + \xi_- - \frac{\eta}{2} )}{b( v_m - \xi_+ + \frac{\eta}{2} )  b( v_m - \xi_- + \frac{\eta}{2} ) } \prod_{i=1}^N \frac{ b(v_m - t_i + \eta ) b( -v_m - t_i - \eta )}{b( v_m - t_i - \eta ) b( -v_m - t_i + \eta )} &=& \prod_{k \neq m} \frac{ b(v_m - v_k + \eta ) b( v_m + v_k + \eta ) }{b (v_m - v_k - \eta ) b( v_m + v_k - \eta )}
\end{eqnarray*}

Indeed, looking at the apparent poles when $u=v_m$ and ignoring the $-v_m$ gives the Bethe equations without extra redundancy.
\begin{eqnarray*}
- \frac{ b( v_m + \xi_+ - \frac{\eta}{2}) }{ b( v_m - \xi_+ + \frac{\eta}{2} )} \frac{b(2v_m - \eta ) \Delta_+ ( v_m)}{\Delta_- (v_m) } &=& \prod_{k \neq m} \frac{ b(v_m - v_k + \eta ) b( v_m + v_k + \eta ) }{b (v_m - v_k - \eta ) b( v_m + v_k - \eta )}\\
\end{eqnarray*}
\begin{eqnarray*}
&-& \frac{ b( v_m + \xi_+ - \frac{\eta}{2} ) }{ b( v_m - \xi_+ + \frac{\eta}{2})} b(2v_m - \eta ) \frac{b(v_m + \xi_- - \frac{\eta}{2} ) \alpha (v_m) \delta (-v_m )}{-b(2v_m -\eta ) b(v_m - \xi_- + \frac{\eta}{2} ) \alpha (-v_m) \delta (v_m)} \\&=& \prod_{k \neq m} \frac{ b(v_m - v_k + \eta ) b( v_m + v_k + \eta ) }{b (v_m - v_k - \eta ) b( v_m + v_k - \eta )}\\
\end{eqnarray*}

The associated eigenvalues of the transfer matrix are as in \cite{Sklyanin88}:

\begin{eqnarray*}
\Lambda (u) &=& \frac{b(2u+\eta )}{b(2u)} b(u+\xi_+ - \frac{\eta}{2} ) b(u+\xi_- \frac{\eta}{2}) \alpha(u) \delta(-u) \prod_{j=1}^m \frac{b(u-v_j-\eta ) b(u + v_j - \eta )}{b(u-v_j)b(u+v_j)}\\
&+& \frac{1}{b(2u)} b(u-\xi_+ + \frac{\eta}{2} ) b(2u-\eta ) b(u - \xi_+ \frac{\eta}{2}) \alpha(-u) \delta(u) \prod_{j=1}^m \frac{b(u-v_j+\eta ) b(u + v_j + \eta )}{b(u-v_j)b(u+v_j)}\\
\end{eqnarray*}

\section{Asymptotic Structure of Solutions to Bethe Equations} \label{section3}

We first seek to count the number of solutions to the Bethe equations deep in this chamber $0< Re(t_1) << Re(t_2) << Re(t_3) \cdots Re(t_N)$. This is to show that there $2^N$ solutions for the collections of $\setof{v_i}$ as desired. In the next section, we consider the associated vectors.

\begin{prop}\label{Div01}
\par Let $Re(t_N) \to + \infty$.\\
\begin{itemize}
\setlength\itemsep{-1em}
\item If the $v_i$ remain finite in this limit, the $\setof{v_i}$ satisfy the Bethe equations for a chain of length $N-1$\\
\item If one of them $v_M$ diverges as $t_N + O(1)$, it has asymptotic behavior of the form $e^{v_M} \to w_M e^{t_N} + o(e^{t_N} ) $ for fixed $w_M$ given below and the remaining rapidities satisfy the Bethe equations for the N-1 length chain.\\
\begin{eqnarray*}
w_M^2 &=& q^2 \frac{ q^{4M} - e^{2 \xi_+ + \xi_-} q^{2N}}{ q^{4M} - e^{2 \xi_+ + \xi_-} q^{2N+4}}\\
\end{eqnarray*}
\item There may also be multiple divergences. In this case there is a decoupling between solving the system for the N-1 length chain and a system for the $w_k$ of the divergences.\\
\end{itemize}

\end{prop}

\begin{proof}.\par
First we take a limit as $Re(t_N)$ goes to $+\infty$ but the $v_i$ are finite. Then the only factor in the Bethe equations which involves $t_N$ is
\begin{equation*}
\frac{b(v_m - t_N + \eta )b(-v_m - t_N - \eta )}{b(v_m - t_N - \eta ) b(-v_m - t_N + \eta )}
\end{equation*}
which tends to 1 so we see the corresponding Bethe equations for a chain of length $N-1$ and still M rapidities.\\

For the second part, assume that $v_M$ goes to $+\infty$ as well and we define $w_M$ so that $e^{v_M-t_N}=w_M$. If we can solve for $w_M$ then we will see how to get rid of one inhomogeneity and one rapidity in the Bethe equations.\\

The equations for $m \neq M$ see the $t_N$ factor go to 1 on the left hand side just as before and the $k=M$ factor go to 1 on the right hand side so we see the $M-1$ equations for a chain of length $N-1$ \footnote{
This is unlike the quasiperiodic case when the left hand side gives a factor of $q^2$ causing the horizantal magnetic field to be modified. \cite{Reshetikhin10}}. The $m=M$ equation then determines $w_M$:

\begin{eqnarray*}
e^{2\xi_+ } q^{-1} e^{2 \xi_-} q^{-1} \frac{w_M q - w_M^{-1}q^{-1}}{w_M q^{-1} - w_M^{-1} q} q^{2N} &=& q^{4(M-1)}\\
w_M^2 &=& q^2 \frac{ q^{4M} - e^{2 \xi_+ + \xi_-} q^{2N}}{ q^{4M} - e^{2 \xi_+ + \xi_-} q^{2N+4}}\\
\end{eqnarray*}

which is nonsingular and nonzero provided that

\begin{eqnarray*}
e^{2 \xi_+ + 2 \xi_- } q^{-1} - q^{4M-2N-5} &=& 0\\
\end{eqnarray*}

so we assume that the parameters $q$ and $\xi_\pm$ are chosen away from this bad locus. This solution is unique up to sign and fixes the behavior of $e^{v_M}$ to be $w_M e^{t_N} + o (e^{t_N} )$.\\

If we only had these possibilities we sould have acheived $2^N$ solutions asymptotically in this sector. However we  do have other solutions when multiple $v_i$ may also diverge with $t_N$. For definiteness, say they are the last J+1 in the list and call this index set $\bold{D}$. As before, parameterize these divergences in the same way: $e^{v_j} = w_j e^{t_N} + o(e^{t_N} )$. In this situation, the Bethe equations for any of the nondivergeing rapidities give the Bethe equations of a chain of length $N-1$ with rapidities $v_1$ through $v_{M-J-1}$ and no dependence on the $w_j$.

Considering the Bethe equations for any of the diverging rapidities gives

\begin{eqnarray*}
e^{2 \xi_+} q^{-1} e^{2 \xi_-} q^{-1} q^{2N} \frac{ w_k q - w_k^{-1} q^{-1} }{ w_k q^{-1} - w_k^{-1} q} &=& q^{4( M - J-1)} q^{2J} \prod_{j \in \bold{D} \; j \neq k} \frac{ b(v_k - t_N + t_N - v_{j} + \eta)}{ b(v_k - t_N + t_N - v_{j} - \eta)}\\
e^{2 \xi_+} q^{-1} e^{2 \xi_-} q^{-1} q^{2N} \frac{ w_k q - w_k^{-1} q^{-1} }{ w_k q^{-1} - w_k^{-1} q} &=& q^{4( M - J-1)} q^{2J} \prod_{j \in \bold{D} \; j \neq k} \frac{w_k w_{j}^{-1} q - w_k^{-1} w_{j} q^{-1}}{w_k w_{j}^{-1} q^{-1} - w_k^{-1} w_{j} q}
\end{eqnarray*}

This gives J quadratic equations used to solve for $w_j \; \forall j \in \bold{D}$. This system does have solutions, but we will see that the associated Bethe vectors are not independent asymptotically. These solutions would contribute to asymptotic completeness in the tensor product of a spin chain with tensorands being Verma modules.

\vspace{-\belowdisplayskip}\[\]
\end{proof}

\section{Asymptotics of Bethe Vectors} \label{section4}

Because the multiple divergences were not excluded from the system of $J+1$ equations above, we must show that the associated vectors are not included. This is done by isolating the dependence of $t_N$ in the Bethe vectors. This then allows us to show that in the case of multiple divergences, these vectors vanish.

\subsection{Isolating the contributions from $t_N$}

First let us change parameterizations to use the following formula of \cite{Reshetikhin13}.

\begin{eqnarray*}
\bar{\mathcal{B}}^{\xi (M)} ( \vec{x}, \vec{t} ) \Omega &=& \sum_{\epsilon=\pm^M} \sum_{\bold{J} \subset \setof{1 \cdots M}} \mathcal{Y}^{\xi, \epsilon , \bold{J}}  (\vec{x}, \vec{t} ) \prod_{i \in \bold{J}^c} B_N ( - \epsilon_i x_i - \frac{\eta}{2} , t_N ) \prod_{j \in \bold{J}} \hat{B} ( - \epsilon_j x_j - \frac{\eta}{2} , \vec{t} ) \Omega\\
\mathcal{Y}^{\xi , \epsilon , \bold{J}} ( \vec{x} , \vec{t} ) &=& \prod_{i=1}^M \big( \epsilon_i b( \xi - \epsilon_i x_i - \frac{\eta}{2} )\prod_{r=1}^N \frac{b( \epsilon_i x_i - t_r - \frac{\eta}{2})}{b( \epsilon_i x_i - t_r + \frac{\eta}{2})} \big)\\
&\times& \prod_{1 \leq i \lt j \leq M} \frac{b( \epsilon_i x_i + \epsilon_j x_j + \eta)}{b( \epsilon_i x_i + \epsilon_j x_j )} Y^\bold{J} ( ( -\epsilon_i x_i - \frac{\eta}{2} ), \vec{t} )\\
Y^\bold{J} ( \vec{x} , \vec{t} ) &=& \prod_{i \in \bold{J}} \frac{b(x_i-t_N)}{b(x_i-t_N + \eta)} \prod_{(i,j) \in \bold{J} \times \bold{J}^c} \frac{b(x_i-x_j+\eta )}{b(x_i-x_j)}
\end{eqnarray*}

where $B_N$ is the operator on just the Nth site and $\hat{B}$ is the matrix element of the double row monodromy matrix with the N'th site omitted removing $t_N$ dependence.\\

With all the $t_N$ dependences now isolated, we may compute asymptotics as $t_N \to + \infty$ for either $\vec{x}$ all remaining finite or some $x_j \to \infty$ in $e^{x_j - t_N} \to w_j$.\\

\subsection{Vanishing for Multiple Divergences}

\begin{theorem}
Zero or one $v_i$ diverging with $t_N$ are the only two linearly independent possibilities. Restricting to these implies the Bethe equations give an asymptotically complete set of solutions.
\end{theorem}

\begin{proof}

Multiple divergences being linearly dependent on the zero or one case is implied by \cref{2DivLinDep,MoreDiv}. We divert those to the next subsection.

Proceed by induction. Suppose for the induction step that we have already given a chain of length $N-1$ there are $\binom{N-1}{M}$ solutions for the M magnon sector. As $t_N$ goes to infinity, the first two items of \cref{Div01} lends $\binom{N}{M}$ solutions of the M magnon sector of the length N chain interpreted as coming from either M magnons on an N-1 chain or M-1 magnons on an N-1 chain by $\binom{N-1}{M} + \binom{N-1}{M-1}$. Adding up all the M sectors gives the desired $2^N$ dimensional Hilbert space.

\vspace{-\belowdisplayskip}\[\]
\end{proof}

\subsection{The Dominant terms}

In fact we may explicitly give the dominant terms for the three cases of zero, one or many diverging $x_i$. This is also necessary to show how even after rescaling leaves a linearly dependent vector in the case of multiple divergences.

\subsubsection{All Finite}

If every $x_i$ remains finite, $\bold{J}$ needs to be $\setof{1 \cdots M}$ because for any other $\bold{J}$ the vector has prefactor $e^{-t_N}$. This leaves a sum of $2^M$ terms all of which do not affect the Nth site with corrections suppressed as $e^{-t_N}$\\

\begin{eqnarray*}
\mathcal{B}^{\xi (M)} ( \vec{x} , \vec{t} ) \Omega &\approx& q^{-M} \prod_{i=1}^M \bigg( \prod_{r=1}^{N-1} \frac{b(x_i - t_r )}{b(x_i - t_r - \eta )} \bigg) \frac{ b(2x_i )}{ b(\xi - x_i - \eta ) b( 2 x_i - \eta ) } \\
&\sum_{\epsilon=\pm^M}& q^{-2M} \mathcal{Y}^{\xi, \epsilon , \setof{1 \cdots M}}  (\vec{x} - \frac{\eta}{2}, \setof{t_1 \cdots t_{N-1}} ) \prod_{j=1}^M \hat{B} ( - \epsilon_j x_j + \epsilon_j \frac{\eta}{2} - \frac{\eta}{2} , \vec{t} ) \Omega\\
&=& q^{-3M} \mathcal{B}^{\xi (M)} ( \vec{x} , \setof{t_1 \cdots t_{N-1}} ) \Omega_{N-1} \otimes e_+\\
\end{eqnarray*}

\subsubsection{Single Diverging Rapidity}\label{SingleDivergence}

For the single divergence $x_M$, there are only contributions from the $J=\setof{1 \cdots M-1}$ summands. The Mth sign is also fixed to be - in this case. This leaves a sum of $2^{M-1}$ terms all of which flip the Nth site.\\

\begin{eqnarray*}
\bar{\mathcal{B}}^{\xi (M)} ( \vec{x}, \vec{t} ) \Omega &\approx& \sum_{\epsilon=\pm^{M-1}-} \mathcal{Y}^{\xi, \epsilon , \setof{1 \cdots M-1}}  (\vec{x}, \vec{t} ) B_N ( x_M - \frac{\eta}{2} , t_N )  \prod_{j=1}^{M-1} \hat{B} ( - \epsilon_j x_j - \frac{\eta}{2} , \vec{t} ) \Omega\\
\mathcal{B}^{\xi (M)} ( \vec{x} + \frac{\eta}{2} , \vec{t} ) \Omega &\approx& \prod_{i=1}^M \bigg( \prod_{r=1}^{N} \frac{b(x_i - t_r + \frac{\eta}{2})}{b(x_i - t_r - \frac{\eta}{2})} \bigg) \frac{ b(2x_i + \eta )}{ b(\xi - x_i - \frac{\eta}{2} ) b( 2 x_i ) } \\
 &\sum_{\epsilon=\pm^{M-1}-}& \mathcal{Y}^{\xi, \epsilon , \setof{1 \cdots M-1}}  (\vec{x}, \vec{t} ) B_N ( x_M - \frac{\eta}{2} , t_N )  \prod_{j=1}^{M-1} \hat{B} ( - \epsilon_j x_j - \frac{\eta}{2} , \vec{t} ) \Omega\\
\end{eqnarray*}

\subsection{Multiple Diverging Rapidities}

If multiple rapidities diverge, the dominant terms follow a similar pattern as one divergence even though in this case the vector goes to $\vec{0}$ as $t_N \to + \infty$. Again for definiteness let us say $\bold{D} = \setof{ J \cdots M}$. We now have the option of which element of $\bold{D}$ to insert into the $B_N$ factor.\\

\begin{eqnarray*}
\mathcal{B}^{\xi (M)} ( \vec{x} , \vec{t} ) \Omega &\approx& \prod_{i=1}^M \bigg( \prod_{r=1}^{N} \frac{b(x_i - t_r)}{b(x_i - t_r - \eta)} \bigg) \frac{ b(2x_i )}{ b(\xi - x_i ) b( 2 x_i - \eta ) } \\
\sum_{\bold{J} = \setof{1 \cdots J \cdots \hat{k} \cdots M}} &\mathlarger{\sum}_{\substack{\epsilon=\pm^{M}\\ \epsilon_k = -1} }& \mathcal{Y}^{\xi, \epsilon , \bold{J}}  (\vec{x} -\frac{\eta}{2}, \vec{t} ) B_N ( x_k - \eta , t_N )  \prod_{j=1,j \neq k}^{M} \hat{B} ( - \epsilon_j x_j + \epsilon_j \frac{\eta}{2} - \frac{\eta}{2} , \vec{t} ) \Omega\\
\end{eqnarray*}

This vector goes to $0$ as $e^{-st_N}$ with $s=\abs{\bold{D}}-1$. We see what happens if we rescale that and compare to the previous two cases of zero or one divergence. We begin with the case of $\abs{D}=2$.

\begin{lemma}[2 Diverging Rapidities\label{2DivLinDep}]
If both $v_{1,2}$ diverge as $e^{v_i} = w_i e^{t_N}$, then $\mathcal{B}^{\xi (2)} ( v_1  , v_2 , \vec{t} ) \Omega$, then there exist a set of $ \setof{z_{2,i}}$ such that $\sum_{z_{2,i}} \mathcal{B}^{\xi (2)} ( v_1  , z_{2,i} , \vec{t} ) \Omega \propto \mathcal{B}^{\xi (2)} ( v_1  , v_2 , \vec{t} ) \Omega$ and the $\setof{z_{2,i}}$ are all finite in the $t_N \to + \infty$ limit.\\
\end{lemma}

\begin{proof}

\begin{eqnarray*}
\mathcal{B}^{\xi (2)} ( v_1  , v_2 , \vec{t} ) \Omega &\approx& q^{2(N-1)} \frac{w_1 - w_1^{-1}}{ w_1 q^{-1} - w_1^{-1} q } \frac{w_2 - w_2^{-1}}{ w_2 q^{-1} - w_2^{-1} q } q^{2} \frac{1}{ w_1 w_2 e^{2 t_N} e^{-2 \xi}}\\
&\mathlarger{\sum}_{\substack{k = 1,2}}& \sum_{\epsilon = \pm }\mathcal{Y}^{\xi, \epsilon , J}  (\vec{x} -\frac{\eta}{2}, \vec{t} ) B_N ( x_k - \eta , t_N )  \prod_{j=1,j \neq k}^{2} \hat{B} ( - \epsilon x_j - \frac{\eta}{2} ( 1 - \epsilon ) , \vec{t} ) \Omega\\
\end{eqnarray*}

Because only linear dependence matters in this section the first line can be ignored except for the $e^{2 t_N}$. 

\begin{eqnarray*}
\mathcal{B}^{\xi (2)} (v_1 , v_2 , \vec{t} ) \Omega &\propto& \frac{-1}{e^{2t_N} } b ( \xi  + v_1 - \eta ) b ( \xi - v_2 ) \frac{ b ( -v_1 - t_N )}{ b( -v_1 + \eta - t_N )}\\
&&  \frac{ b( -v_2 + v_1 - \eta )}{ b( -v_2 + v_1 )} \frac{ b( v_2 - \eta - t_N )}{ b( v_2 - t_N )} B_N ( v_1 - \eta, t_N) \hat{B} ( - v_2 , \vec{t} ) \Omega\\
&+& \frac{1}{e^{2t_N}} b ( \xi  + v_1 - \eta ) b( \xi + v_2 - \eta )  q^{2(N-1)} \frac{ b( - v_1 - t_N )}{ b( -v_1 + \eta - t_N )}\\ && \frac{  b( - v_2 + v_1 - \eta )}{ b( v_1 - v_2 )} \frac{ b( v_2 - \eta - t_N )}{ b( v_2 - t_N )} B_N ( v_1 - \eta , t_N ) \hat{B} ( v_2 - \eta , \vec{t} ) \Omega\\
&+& \frac{-1}{e^{2t_N} } b ( \xi  + v_2 - \eta ) b ( \xi - v_1 ) \frac{ b ( -v_2 - t_N )}{ b( -v_2 + \eta - t_N )}\\
&&  \frac{ b( -v_1 + v_2 - \eta )}{ b( -v_1 + v_2 )} \frac{ b( v_1 - \eta - t_N )}{ b( v_1 - t_N )} B_N ( v_2 - \eta, t_N) \hat{B} ( - v_1 , \vec{t} ) \Omega\\
&+& \frac{1}{e^{2t_N}} b ( \xi  + v_2 - \eta ) b( \xi + v_1 - \eta )  q^{2(N-1)} \frac{ b( - v_2 - t_N )}{ b( -v_2 + \eta - t_N )}\\ && \frac{  b( - v_1 + v_2 - \eta )}{ b( v_2 - v_1 )} \frac{ b( v_1 - \eta - t_N )}{ b( v_1 - t_N )} B_N ( v_2 - \eta , t_N ) \hat{B} ( v_1 - \eta , \vec{t} ) \Omega\\
\end{eqnarray*}

This simplifies upon defining an auxiliary $R$ $S_1$ and $S_2$ to the following: 

\begin{eqnarray*}
R &\equiv& \frac{ b( -v_1 - t_N ) }{ b( -v_1 + \eta - t_N )} \frac{ b ( -v_2 + \eta - t_N )}{ b( -v_2 - t_N )} \frac{ b( v_1 - v_2 - \eta )}{ b( v_2 - v_1 - \eta )}\\
&& \frac{ b ( v_2 - \eta - t_N )}{ b( v_2 - t_N ) } \frac{ b( v_1 - t_N )}{ b( v_1 - \eta - t_N )} \frac{ b( v_2 - t_N )}{ b( v_1 - t_N )} e^{v_1 - v_2 }\\
&\approx& \frac{ b( -v_1 - t_N ) }{ b( -v_1 + \eta - t_N )} q^{-1} \frac{ b( v_1 - v_2 - \eta )}{ b( v_2 - v_1 - \eta )} \frac{ b ( v_2 - \eta - t_N )}{ b( v_1 - \eta - t_N )} e^{v_1 - v_2 }\\
S_1 &\equiv& \frac{-1}{e^{2t_N} } b ( \xi  + v_1 - \eta ) b ( \xi - v_2 ) \frac{ b ( -v_1 - t_N )}{ b( -v_1 + \eta - t_N )}\\
&&  \frac{ b( -v_2 + v_1 - \eta )}{ b( -v_2 + v_1 )} \frac{ b( v_2 - \eta - t_N )}{ b( v_2 - t_N )} B_N ( v_1 - \eta, t_N) \hat{B} ( - v_2 , \vec{t} ) \Omega\\
S_2 &\equiv& \frac{1}{e^{2t_N}} b ( \xi  + v_1 - \eta ) b( \xi + v_2 - \eta )  q^{2(N-1)} \frac{ b( - v_1 - t_N )}{ b( -v_1 + \eta - t_N )}\\ && \frac{  b( - v_2 + v_1 - \eta )}{ b( v_1 - v_2 )} \frac{ b( v_2 - \eta - t_N )}{ b( v_2 - t_N )} B_N ( v_1 - \eta , t_N ) \hat{B} ( v_2 - \eta , \vec{t} ) \Omega\\
\mathcal{B}^{\xi (2)} ( v_1  , v_2 , \vec{t} ) \Omega &\propto& (1 + R) ( S_1 + S_2 )\\
\end{eqnarray*}

This overall prefactor $(1+R)$ can be dropped leaving.

\begin{eqnarray*}
S_1 + S_2 &=&  \frac{-1}{e^{2t_N} } \frac{ b ( -v_1 - t_N )}{ b( -v_1 + \eta - t_N )} \frac{ b( -v_2 + v_1 - \eta )}{ b( -v_2 + v_1 )} \frac{ b( v_2 - \eta - t_N )}{ b( v_2 - t_N )} B_N ( v_1 - \eta, t_N) \\ &&  b ( \xi  + v_1 - \eta ) b ( \xi - v_2 ) \hat{B} ( - v_2 , \vec{t} ) \Omega\\
&+&  \frac{-1}{e^{2t_N}} \frac{ b( - v_1 - t_N )}{ b( -v_1 + \eta - t_N )} \frac{  b( - v_2 + v_1 - \eta )}{ b( v_1 - v_2 )} \frac{ b( v_2 - \eta - t_N )}{ b( v_2 - t_N )} B_N ( v_1 - \eta , t_N )\\ && - b ( \xi  + v_1 - \eta ) b( \xi + v_2 - \eta )  q^{2(N-1)}  \hat{B} ( v_2 - \eta , \vec{t} ) \Omega\\
&\propto& B_N ( v_1 - \eta , t_N ) ( \hat{B} ( -v_2 -  N \eta, \vec{t} ) - \hat{B} ( v_2 - N \eta , \vec{t} ) ) \Omega\\
( \hat{B} ( -v_2 -  N \eta, \vec{t} ) - \hat{B} ( v_2 - N \eta , \vec{t} ) ) \Omega &\approx& -\sum_{r=0}^{N-1} q^{-r} \frac{q-q^{-1}}{1} (\frac{1}{e^{v_2 - N \eta -t_{r+1}+\eta}} + \frac{1}{e^{v_2 + N \eta + t_{r+1} - \eta }} )  e_+^{r} e_- e_+^{N-1-r}\\
&\approx& - \frac{q-q^{-1}}{e^{v_2}} \sum_{r=0}^{N-1} q^{-r} (\frac{1}{e^{- N \eta -t_{r+1}+\eta}} + \frac{1}{e^{N \eta + t_{r+1} - \eta }} )  e_+^{r} e_- e_+^{N-1-r}\\
&=& - \frac{q-q^{-1}}{e^{v_2}} \sum_{r=0}^{N-1} q^{-r} 2 \cosh ( N \eta + t_{r+1} - \eta )  e_+^{r} e_- e_+^{N-1-r}\\
\end{eqnarray*}

Compare this with the vectors we already have in the single divergence sector where b and c are of the form $t_N + O(1)$ with constant term to be determined.

\begin{eqnarray*}
\mathcal{B} ( z_2 , v_1 , \vec{t} ) \Omega &=& B_N(v_1 - \eta , t_N ) \sum_{r=0}^{N-1} \bigg( \mathcal{Y}^{ \xi , (-1,-1) , \setof{1} } ( (z_2 , v_1 ) - \frac{\eta}{2} , \vec{t} ) \prod_{i=1}^r \big( \frac{ b( z_2 - 2 \eta - t_i )}{b( z_2 - \eta - t_i )} \big) \frac{ b(\eta )}{ b( z_2 - 2 \eta - t_{r+1} )} \\
&+&  \mathcal{Y}^{ \xi , (-1,+1) , \setof{1} } ( (z_2 , v_1 ) - \frac{\eta}{2} , \vec{t} ) \prod_{i=1}^r \big( \frac{ b( - z_2 - \eta - t_i )}{b( - z_2 - t_i )} \big) \frac{ b(\eta )}{ b( - z_2 - \eta - t_{r+1} )}
\bigg) e_+^r e_- e_+^{N-1-r}\\
\end{eqnarray*}

We seek to show that

\begin{eqnarray*}
\mathcal{B}^{\xi (2)} ( v_1  , v_2 , \vec{t} ) \Omega &=& \sum_{\alpha} \mathcal{B}^{\xi (2)} ( z_{2,\alpha} , v_1  , \vec{t} ) \Omega
\end{eqnarray*}

for some set of regular $z_{2,\alpha}$. Matching coefficients gives the system of equations for all $r$:

\begin{eqnarray*}
e^{-t_N} \sum_{\alpha} \bigg( \mathcal{Y}^{ \xi , (-1,-1) , \setof{1} } ( (z_{2,\alpha} , v_1 ) - \frac{\eta}{2} , \vec{t} ) &\prod_{i=1}^r& \big( \frac{ b( z_{2,\alpha} - 2 \eta - t_i )}{b( z_{2,\alpha} - \eta - t_i )} \big) \frac{ b(\eta )}{ b( z_{2,\alpha} - 2 \eta - t_{r+1} )} \\
+  \mathcal{Y}^{ \xi , (-1,+1) , \setof{1} } ( (z_{2,\alpha} , v_1 ) - \frac{\eta}{2} , \vec{t} ) &\prod_{i=1}^r& \big( \frac{ b( - z_{2,\alpha} - \eta - t_i )}{b( - z_2 - t_i )} \big) \frac{ b(\eta )}{ b( - z_{2,\alpha} - \eta - t_{r+1} )} \bigg)\\
&=& - \frac{q-q^{-1}}{e^{v_2}} q^{-r} 2 \cosh ( N \eta + t_{r+1} - \eta )\\
\end{eqnarray*}

Solving this equation for the set of $z_{2,\alpha}$ then shows the desired linear dependence.

\vspace{-\belowdisplayskip}\[\]
\end{proof}

Now proceeding by induction for the rest when $\abs{\bold{D}} > 2$ results in:

\begin{lemma}[$\geq 2$ Diverging Rapidities\label{MoreDiv}]
Any state $e^{s t_N} \mathcal{B} ( \cdots , \vec{t} ) \Omega$ where $s \geq 0$ is one less than the number of divergences ( This ensures that this vector has finite nonzero norm in the $t_N \to +\infty$ limit.) can be approximated by a linear combination of states of the form \ref{SingleDivergence}
\end{lemma}

\begin{proof}

By \cref{2DivLinDep} we have a base case.

\begin{eqnarray*}
a e^{2 t_N} \mathcal{B} (x_1 + t_N , x_2 + t_N , \vec{t} ) \Omega &=& \sum b_i \mathcal{B} ( y_{1i} + t_N , y_{2i} , \vec{t} ) \Omega\\
\end{eqnarray*}

Therefore when including other rapidities, they come along for the ride as:

\begin{eqnarray*}
\mathcal{B} ( x_1 + t_N , x_2 + t_N , \cdots  , \vec{t} ) \Omega &=& \mathcal{B} ( \cdots , \vec{t} ) a^{-1} e^{-2 t_N} \sum b_i \mathcal{B} ( y_{1i} + t_N , y_{2i} , \vec{t} ) \Omega\\
e^{2 t_N} \mathcal{B} ( x_1 + t_N , x_2 + t_N , \cdots , \vec{t} ) \Omega &=& \sum a^{-1} b_i \mathcal{B} ( \cdots , y_{1i} + t_N ,y_{2i}, \vec{t} ) \Omega
\end{eqnarray*}

Now take each $\mathcal{B} ( \cdots , y_{1i} + t_N ,y_{2i}, \vec{t} ) \Omega$ in the RHS and repeat the procedure if there are still $J\geq 2$ diverging rapidities. The overall factor for rescaling is $e^{2 t_N}$ for each extra divergence. The base case requires two divergences so we can reduce to the single divergence case as described in \ref{SingleDivergence} and no further.

\vspace{-\belowdisplayskip}\[\]
\end{proof}

\section{Conclusion}
We have shown that as each inhomogeneity is taken to infinity the solutions to the Bethe equations break up as all remaining finite or one going to infinity in a prescribed manner. The lack of other possibilities gives the desired completeness property. This was done by looking at the asymptotics of the Bethe equations to produce the different solution sets followed by a check of linear dependence in the finite dimensional quotient. So what appear to be extra on shell vectors are actually only independent in the Verma, but not in the quotient. \par
There are more general solutions of the reflection equation which are not diagonal. These requires use of the Dynamical Yang Baxter Equation after the ``gauge transformation." Other problems include specialization to combinatorial points. We may also consider these sorts of limits as they arise in the context of defects whereas the original Hamiltonian comes from the limiting behavior as $t_i \to 1$. This sort of $t_i \to +\infty$ limits appear when taking successively larger spins $V(j,t)$ which are built from fusion of many $V(1/2,t_i)$. So even if $t \to 1$, some of the $t_i \to + \infty$\cite{BqKZFusion,HernandezJimbo}. These are left for future work.

\bibliography{ReflectingBC12}
\bibliographystyle{alpha}

\end{document}